\documentclass[encountsame]{llncs}
\newcommand{\uRule}{\ensuremath{\forall}-rule\xspace}
\newcommand{\eRule}{\ensuremath{\exists}-rule\xspace}
\newcommand{\aRule}{\ensuremath{\approx}-rule\xspace}
\newcommand{\eqRule}{\textit{Eq}-rule\xspace}

\newcommand{\egdRules}[1]{\ensuremath{#1^\approx}\xspace}
\newcommand{\eR}[1]{\egdRules{#1}}

\newcommand{\overchase}[1]{\ensuremath{\mathcal{V}_#1}\xspace}

\newcommand{\V}{\ensuremath{\mathcal{V}}\xspace}

\newcommand{\vx}{\ensuremath{\vec{x}}\xspace}
\newcommand{\vy}{\ensuremath{\vec{y}}\xspace}
\newcommand{\vz}{\ensuremath{\vec{z}}\xspace}
\newcommand{\vt}{\ensuremath{\vec{t}}\xspace}
\newcommand{\vs}{\ensuremath{\vec{s}}\xspace}

\newcommand{\formula}{\ensuremath{\phi}\xspace}
\newcommand{\fact}{\ensuremath{\phi}\xspace}
\newcommand{\query}{\ensuremath{\gamma}\xspace}
\newcommand{\auxRule}{\ensuremath{\upsilon}\xspace}

\newcommand{\EGD}{\ensuremath{\body(\vx) \to x \approx y}\xspace}
\newcommand{\body}{\ensuremath{\beta}\xspace}
\newcommand{\head}{\ensuremath{\eta}\xspace}

\newcommand{\skolemFunction}[2]{\ensuremath{f_{#2}^{#1}}}
\newcommand{\sF}[2]{\skolemFunction{#1}{#2}}
\newcommand{\sk}{\textit{sk}}

\newcommand{\Program}{\ensuremath{\mathcal{P}}\xspace}
\renewcommand{\P}{\Program}
\newcommand{\Rules}{\ensuremath{\mathcal{R}}\xspace}
\newcommand{\R}{\Rules}
\newcommand{\Instance}{\ensuremath{\mathcal{I}}\xspace}
\newcommand{\I}{\Instance}
\newcommand{\RI}{{\ensuremath{\langle \R,\I \rangle}}\xspace}
\newcommand{\PRI}{{\ensuremath{\P = \langle \R,\I \rangle}}\xspace}
\newcommand{\Facts}{\ensuremath{\mathcal{F}}\xspace}
\newcommand{\F}{\Facts}

\newcommand{\depth}{\textit{dep}\xspace}
\newcommand{\substitution}{\ensuremath{\sigma}\xspace}
\newcommand{\subs}{\substitution}

\newcommand{\strictAbitraryOrder}{\ensuremath{\leadsto}\xspace}
\newcommand{\sAO}{\strictAbitraryOrder}

\newcommand{\obliviousChaseStep}[2]{\ensuremath{#1_{O}^{#2}}\xspace}
\newcommand{\oCS}[2]{\obliviousChaseStep{#1}{#2}}
\newcommand{\obliviousChase}[1]{\ensuremath{\textit{OC}(#1)}}
\newcommand{\oChase}[1]{\obliviousChase{#1}}
\newcommand{\restrictedChaseStep}[2]{\ensuremath{#1_{R}^{#2}}\xspace}
\newcommand{\rCS}[2]{\restrictedChaseStep{#1}{#2}}
\newcommand{\restrictedChase}[1]{\ensuremath{\textit{RC}(#1)}}
\newcommand{\rChase}[1]{\restrictedChase{#1}}

\newcommand{\criticalInstance}[1]{\ensuremath{\mathcal{I}_\star(#1)}\xspace}
\newcommand{\cI}[1]{\criticalInstance{#1}}

\newcommand{\hsriq}{\textit{Horn-}\ensuremath{\mathcal{SRIQ}}\xspace}

\newcommand{\Nc}{\ensuremath{\textit{N}_\textit{C}}\xspace}
\newcommand{\Ni}{\ensuremath{\textit{N}_\textit{I}}\xspace}
\newcommand{\Nr}{\ensuremath{\textit{N}_\textit{R}}\xspace}
\newcommand{\Nri}{\ensuremath{\textit{N}^-_\textit{R}}\xspace}

\newcommand{\atMostOne}[2]{{\leq}\, 1\, #1 . #2}
\newcommand{\aMO}[2]{\atMostOne{#1}{#2}}

\newcommand{\ontology}{\ensuremath{\mathcal{O}}\xspace}
\renewcommand{\O}{\ontology}
\newcommand{\TBox}{\ensuremath{\mathcal{T}}\xspace}
\newcommand{\T}{\TBox}
\newcommand{\existsAxioms}[1]{\ensuremath{#1^\exists}\xspace}
\newcommand{\eA}[1]{\existsAxioms{#1}}
\newcommand{\univAxioms}[1]{\ensuremath{#1^\forall}\xspace}
\newcommand{\uA}[1]{\univAxioms{#1}}
\newcommand{\ABox}{\ensuremath{\mathcal{A}}\xspace}
\newcommand{\A}{\ABox}
\newcommand{\TA}{\ensuremath{\langle \TBox, \ABox \rangle}\xspace}
\newcommand{\OTA}{\ensuremath{\O = \TA}\xspace}

\newcommand{\equalitySet}{\ensuremath{\textsf{Eq}}\xspace}

\newcommand{\eP}{\textit{Eq}\xspace}

\newcommand{\Rt}{\ensuremath{\R_\T}\xspace}

\newcommand{\Po}{\ensuremath{\P(\O)}\xspace}
\newcommand{\Pta}{\ensuremath{\P(\TA)}\xspace}

\newcommand{\RCA}[1]{\ensuremath{\text{RCA}_{#1}}}

\newcommand{\RTerm}[2]{\ensuremath{\mathcal{U}(#1, #2)}\xspace}
\newcommand{\RT}[2]{\RTerm{#1}{#2}}

\newcommand{\MFA}{\text{MFA}\xspace}
\newcommand{\MFAE}{\ensuremath{\text{MFA}^\exists}\xspace}
\newcommand{\MFAF}{\ensuremath{\text{MFA}^\forall}\xspace}
\newcommand{\MFAC}{\ensuremath{\text{MFA}^\cup}\xspace}

\newcommand{\ExpTime}{\textsc{ExpTime}\xspace}

\newcommand{\ExpTimeC}{\textsc{ExpTime-}complete\xspace}

\newcommand{\Film}{\textit{Film}}

\newcommand{\isProducedBy}{\textit{isProdBy}}

\newcommand{\Producer}{\textit{Producer}}

\newcommand{\produces}{\textit{prod}}

\newcommand{\AI}{\textit{AI}}

\newcommand{\Naturals}{\ensuremath{\mathbb{N}}\xspace}

\usepackage{amsmath}
\usepackage{amssymb}
\usepackage{enumerate}
\usepackage{xspace}
\usepackage{tabu}
\usepackage{url}
\usepackage{multirow}

\begin{document}

\title{A Practical Acyclicity Notion for Query Answering over \hsriq{} Ontologies\protect\footnote[1]{This is a post-peer-review, pre-copyedit version of an article published at the 15th International Semantic Web Conference (ISWC 2016). The final authenticated version is available online at:  \protect\url{http://dx.doi.org/10.1007/978-3-319-46523-4_5}”.}}
 
\author{David Carral$^1$ \and Cristina Feier$^2$ \and Pascal Hitzler$^1$}
\institute{DaSe Lab, Wright State University, Dayton US \and Universit\"at Bremen, Bremen Germany}

\maketitle

\begin{abstract}
Conjunctive query answering over expressive Horn Description Logic ontologies is a relevant and challenging problem which, in some cases, can be addressed by application of the chase algorithm.
In this paper, we define a novel acyclicity notion which provides a sufficient condition for termination of the restricted chase over \hsriq TBoxes.
We show that this notion generalizes most of the existing acyclicity conditions (both theoretically and empirically). Furthermore, this new acyclicity notion gives rise to a very efficient reasoning procedure. We provide evidence for this by providing a materialization based reasoner for acyclic ontologies which outperforms other state-of-the-art systems.
\end{abstract}

\section{Introduction}
\label{section:introduction}

Conjunctive query (CQ) answering over expressive Description Logic (DL) ontologies is a key reasoning task which remains unsolved for many practical purposes.
Indeed, answering CQs over DL ontologies is quite intricate and often of high computational complexity \cite{Calvanese201412,DBLP:journals/jair/GlimmLHS08,DBLP:journals/corr/RudolphG14}.
Nevertheless, CQ answering over a major class of DLs, the so-called \emph{Horn DLs}, can in some cases be addressed via application of the \emph{chase algorithm}, a technique where all relevant consequences of an ontology are precomputed, allowing queries to be directly evaluated over the materialized set of facts.
However, the chase is not guaranteed to terminate for all ontologies, and checking whether it does is not a straightforward procedure.
It is thus an ongoing research endeavor to establish so-called \emph{acyclicity conditions}; i.e., sufficient conditions which ensure termination of the chase.

The main contribution of this paper is the definition of \emph{restricted chase acyclicity} (\RCA{n}), a novel acyclicity condition for \hsriq ontologies (the DL \hsriq may be informally described as the logic underpinning the deterministic fragment of OWL DL \cite{owl2-primer} minus nominals).
If an ontology is proven to be \RCA{n}, then $n$-cyclic terms do not occur during the computation of the chase of such ontology and thus the chase is guaranteed to terminate.

In contrast with existing acyclicity notions \cite{CG+13:acyclicity} which deal with termination of the unrestricted, i.e. oblivious, chase of arbitrary sets of existential rules, we restrict our attention  to the language \hsriq and seek to achieve termination of the restricted chase algorithm \cite{CaliGK08}; this is a special variant of the standard chase in which the inclusion of further terms to satisfy existential restrictions is avoided if such restrictions are already satisfied, and equality is dealt with via renaming.
By considering such a chase algorithm we are able to devise acyclicity conditions which are more general than any other of the notions previously described.

On the theoretical side, we show that \RCA{n} is more general than \emph{model-faithful acyclicity} (\MFA) provided $n$ is sufficiently large (linear in the size of ontology).
As shown in \cite{CG+13:acyclicity}, this is one of the most general acyclicity conditions for ontologies described to date, as it encompasses many other existing notions such as \emph{joint acyclicity} \cite{DBLP:conf/ijcai/KrotzschR11}, \emph{super-weak acyclicity} \cite{DBLP:conf/pods/Marnette09} or the hybrid acyclicity notions presented in \cite{DBLP:conf/ecai/BagetGMR14}.
Furthermore, we show that deciding \RCA{n} membership is not harder than deciding \MFA membership.

On the practical side, we empirically show that (i) \RCA{n} characterizes more real-world ontologies as acyclic than \MFA.
Furthermore, we demonstrate that (ii) the specific type of acyclicity captured by \RCA{n} results in a more efficient reasoning procedure.
This is because acyclicity is still preserved in the case when employing renaming techniques when reasoning in the presence of equality.
Thus, the use of cumbersome axiomatizations of equality such as \emph{singularization} \cite{DBLP:conf/pods/Marnette09} can be avoided.
Moreover, we report on an implementation of the restricted chase algorithm based on the datalog engine RDFOx \cite{DBLP:conf/semweb/NenovPMHWB15} and show that (iii) it vastly outperforms state-of-the-art DL reasoners.
To verify (i-iii), we complete an extensive evaluation with very encouraging results.

The rest of the paper is structured as follows: We start with some preliminaries in Section \ref{section:preliminaries}. Section \ref{section:reasoningChase} formally introduces the notions of oblivious and restricted chase, followed by an overview of \MFA in Section \ref{section:acyclicityNotions}. In Section \ref{section:RCA} we introduce our new acyclicity notion \RCA{n}. Finally, Section \ref{section:evaluation} and Section \ref{conclusions} describe the evaluation of our work and list our conclusions, respectively.

An extended technical report for this paper with all the proofs and further information concerning the evaluation can be found at \url{http://dase.cs.wright.edu/publications/acyclicity-notion-cqa-over-horn-sriq-ontologies}.
\section{Preliminaries}
\label{section:preliminaries}

\subsubsection{Rules}

We use the standard notions of \emph{constants}, \emph{function symbols} and \emph{predicates}, where $\approx$ is the equality predicate, $\top$ is universal truth, and $\bot$ is universal falsehood.
\emph{Variables}, \emph{terms}, \emph{atoms} and \emph{substitutions} are defined as usual.
A \emph{fact} is a ground atom; i.e., an atom without occurrences of variables.
As customary, every term $t$ is associated with some \emph{depth} $\depth(t) \geq 0$.
Furthermore, we often abbreviate a vector of terms $t_1, \ldots, t_n$  as $\vt$ and identify $\vt$ with the set $\{t_1, \ldots, t_n\}$.
In a similar manner, we often identify a conjunction of atoms $\fact_1 \wedge \ldots \wedge \fact_n$ with the set $\{\fact_1, \ldots, \fact_n\}$.
With $\formula(\vx)$ we stress that $\vx = x_1, \ldots, x_n$  are the free variables occurring in the formula \formula.

Let $t$ be some ground term and $c$ some constant.
Let $t_c$ be the term obtained from $t$ by replacing every occurrence of a constant by $c$, i.e., $f(d, g(e))_c = f(c, g(c))$.
The notation is analogously extended to facts and sets of facts.

A term $t'$ is a \emph{subterm} of another term $t$ if and only if $t' = t$, or $t = f(\vs)$ and $t'$ is a subterm of some $s \in \vs$; if additionally $t' \neq t$, then $t'$ is a \emph{proper subterm} of $t$.
A term $t$ is \emph{$n$-cyclic} if and only if there exists a sequence of terms of the form $f(\vec{s_1}), \ldots, f(\vec{s_{n+1}})$ such that  $f(\vec{s_{n+1}})$ is a subterm of $t$ and, for every $i = 1, \ldots, n$, $f(\vec{s_i})$ is a proper subterm of $f(\vec{s_{i+1}})$.
We simply refer to $1$-cyclic terms as \emph{cyclic}.

A \emph{rule} is a first-order logic (FOL) formula of one of the forms
\begin{align}
\forall \vx \forall \vz [\body(\vx, \vz) &\to \exists \vy \head(\vx, \vy)] \label{tgdRule}\qquad \text{or} \\
\forall \vx [\body(\vx) &\to x \approx y] \label{egdRule},
\end{align}
where \body and \head are non-empty conjunctions of atoms which do not contain occurrences of constants, function symbols nor of the predicate $\approx$; \vx, \vy and \vz are pairwise disjoint; and $x, y \in \vx$.
To simplify the notation, we frequently omit the universal quantifiers from rules.
As customary, we refer to rules of the forms (\ref{tgdRule}) and (\ref{egdRule}) as \emph{tuple generating dependencies} (\emph{TGDs}) and \emph{equality generating dependencies} (\emph{EGDs}), respectively.

Given a set of rules \R, we define \eA{\R} and \uA{\R} as the sets of all the TGDs in \R which do and do not contain existentially quantified variables, respectively.
Moreover, let \eR{\R} be the set of all EGDs in \R.
A \emph{program} is a tuple \RI where \R is a set of rules and \I is an \emph{instance}; i.e., a finite set of equality-free facts.

The main reasoning task we are investigating in this paper is CQ answering.
Nevertheless, for the rest of the paper, we restrict our attention to the simpler task of CQ entailment of \emph{boolean conjunctive queries} (BCQs).
This is without loss of generality since CQ answering can be reduced to checking entailment of BCQs.
A \emph{BCQ}, or simply a \emph{query}, is a formula of the form $\exists \vy \head(\vy)$ where \head is a conjunction of atoms not containing occurrences of constants, function symbols nor $\approx$.

For the remainder of the paper, we assume that $\top$ and $\bot$ are treated as ordinary unary predicates and that the semantics of $\top$ is captured explicitly in any program \PRI by including the rule $p(x_1, \ldots, x_n) \to \top(x_1) \wedge \ldots \wedge \top(x_n)$ in $\R$ for every predicate $p$ with arity $n$ occurring in \P.

We interpret programs under standard FOL semantics with true equality.
As usual, a program \P is \emph{satisfiable} if and only if $\P \not \models \exists y\bot(y)$.
Furthermore, given some query \query, we write $\P \models \query$ to indicate that \P \emph{entails} \query.

We will later employ skolemization to define the consequences of a TGD over a set of facts.
The \emph{skolemization} $\sk(\rho)$ of some TGD $\rho = \body(\vx, \vz) \to \exists \vy \head(\vx, \vy)$ is the rule $\body(\vx, \vz) \to \head(\vx, \vy)\subs_{\sk}$ where $\subs_{\sk}$ is a substitution mapping every $y \in \vy$ into $\sF{y}{\rho}(\vx)$ where $\sF{y}{\rho}$ is a fresh function unique for every variable $y$ and TGD $\rho$.

\subsubsection{Description Logics}

We next define the syntax and semantics of the ontology language \hsriq \cite{KRH:HornDLs2013}.
We assume basic familiarity with DL, and refer the reader to the literature for further details \cite{dlhandbook}.
Without loss of generality, we restrict our attention to ontologies in a normal form close to the one from \cite{KRH:HornDLs2013}.

A \emph{DL signature} is a tuple $\langle \Nc, \Nr, \Ni \rangle$ where $\Nc$, $\Nr$ and $\Ni$ are infinite countable and mutually disjoint sets of \emph{concept names}, \emph{role names} and \emph{individuals}, respectively, such that $\{\bot, \top\} \subseteq \Nc$.
A \emph{role} is an element of $\Nri = \Nr \cup \{R^- \mid R \in \Nr\}$.
\begin{figure*}[t]
\centering
\begin{align*}
A_1 \sqcap \ldots \sqcap A_n &\sqsubseteq B	&~~~\mapsto~~~ &A_1(x) \wedge \ldots \wedge A_n(x) \to B(x) \\
A &\sqsubseteq \forall R.B					&~~~\mapsto~~~ &A(x) \wedge R(x, y) \to B(y) \\
A &\sqsubseteq \aMO{R}{B}				&~~~\mapsto~~~ &A(x) \wedge R(x, y) \wedge B(y) \wedge R(x, z) \wedge B(z) \to y \approx z \\
A &\sqsubseteq \exists R.B				&~~~\mapsto~~~ &A(x) \to  \exists y [R(x, y) \wedge B(y)] \\
S &\sqsubseteq R						&~~~\mapsto~~~ &S(x, y) \to R(x, y) \\
S^- &\sqsubseteq R						&~~~\mapsto~~~ &S(y, x) \to R(x, y) \\
S \circ V &\sqsubseteq R					&~~~\mapsto~~~ &S(x, y) \wedge V(y, z) \to R(x, z)
\end{align*}
\caption{Mapping axioms $\alpha$ to rules $\Pi(\alpha)$, where $A_{(i)}, B \in \Nc$, $R, S, V \in \Nr$.}
\label{figure:Pi}
\end{figure*} 
A \emph{TBox axiom} is a formula of one of the forms given on the left hand side of the mappings in Figure \ref{figure:Pi}.
TBox axioms of the form $A \sqsubseteq \exists R.B$ are also referred as \emph{existential axioms}.
An \emph{ABox axiom} is a formula of the form $A(a)$ or $R(a, b)$ where $A \in \Nc$, $R \in \Nr$ and $a, b \in \Ni$.
An \emph{axiom} is either a TBox or an ABox axiom.
As usual, we simply refer to a set of TBox (resp. ABox) axioms as a \emph{TBox} (resp. an \emph{ABox}).


A \emph{\hsriq ontology} \O (or simply an \emph{ontology}) is some tuple \TA, where \T and \A are a TBox and an ABox, respectively, which satisfies the usual conditions \cite{DBLP:conf/kr/HorrocksKS06}.

Due to the close correspondence between ontologies and programs, we define the semantics of the former by means of a mapping into the latter.
Given some TBox \T, let $\Rt = \Pi(\T)$.
Given some ontology \OTA, let $\Po = (\Rt, \A)$ where $\Pi$ is the function from Figure \ref{figure:Pi}.
We say that \O is \emph{satisfiable} if and only if the program \Po is satisfiable.
Furthermore, \O \emph{entails} a query \query, written $\O \models \query$, if and only if \Po is unsatisfiable or $\Po$ entails $\query$.


\section{The Chase Algorithm}
\label{section:reasoningChase}

In this section we present two variants of the chase algorithm, which are somewhat similar to the oblivious and restricted chase from \cite{CaliGK08}, and elaborate about how such procedures may be used to solve CQ entailment over ontologies.

\begin{definition}
\label{definition:tgdConsequences}
A fact \fact is an \emph{oblivious consequence} of a TGD $\rho = \body(\vx, \vz) \to \exists \vy \head(\vx, \vy)$ on a set of facts \F if and only if there is some substitution \subs with $\body(\vx, \vz)\subs \subseteq \F$ and $\fact \in \sk(\head(\vx, \vy))\subs$ where $\sk(\head(\vx, \vy))$ is the head of the (skolemized) TGD $\sk(\rho)$.
A fact \fact is a \emph{restricted consequence} of $\rho$ on \F if and only if there is a substitution \subs with (1) $\body(\vx, \vz)\subs \subseteq \F$ and $\fact \in \sk(\head(\vx, \vy))\subs$, and (2) there is no substitution $\tau \supseteq \subs$ with $\head(\vx, \vy) \tau \subseteq \F$.

The result of \emph{obliviously applying} $\rho$ to \F, written $\rho_O(\F)$, is the set of all \emph{oblivious consequences} of $\rho$ on \F.
The result of \emph{obliviously applying} a set of TGDs \R to \F, written $\R_O(\F)$, is the set $\bigcup_{\rho \in \R} \rho_O(\F) \cup \F$.
The result of \emph{restrictively applying} $\rho$ to \F (resp., \R to \F), written $\rho_R(\T)$ (resp., $\R_R(\T)$), is analogously defined.
\end{definition}

\begin{definition}
\label{definition:egdConsequences}
Let $\sAO$ be some total strict order over the set of all terms such that $t \sAO u$ only if $\depth(t) \leq \depth(u)$.
Furthermore, we say that $t$ is greater than $u$ with respect to \sAO to indicate $t \sAO u$.

Given a set of EGDs \R and a set of facts \F, let $\mapsto^\R_\F$ be the minimal congruence relation over terms such that $t \mapsto^\R_\F u$ if and only if there exists some $\EGD \in \R$ and some substitution \subs with $\body(\vx)\sigma \subseteq \F$, $\sigma(x) = t$ and $\sigma(y) = u$.
Let $\R(\F)$ be the set that is obtained from \F by replacing all occurrences of every term $t$ by $u$ where $u$ is the greatest term with respect to \sAO such that $t \mapsto^\R_\F u$.
\end{definition}

Note that we define 
consequences with respect to sets of rules instead of simply (single) rules as it is customary \cite{CaliGK08}.
This allows us to define the chase as a deterministic procedure (modulo \sAO).
Also, unlike in \cite{CaliGK08}, where a lexicographic order is used to direct the replacement of terms, we employ a type of order which ensures that terms are always replaced by terms of equal or lesser depth.
This effectively precludes some ``deeper'' terms from being introduced during the computation of the chase.

\begin{definition}
\label{definition:chaseSequence}
Let \PRI be some program.
The \emph{oblivious chase sequence} of \P is the sequence $\F_0, \F_1, \ldots$ such that $\F_1 = \I$ and, for all $i \geq 1$, $\F_i$ is the set of facts defined as follows.
\begin{itemize}
\item If $\R^\approx(\F_{i-1}) \neq \F_{i-1}$, then $\F_i = \R^\approx(\F_{i-1})$.
\item If $\F_{i-1} = \R^\approx(\F_{i-1})$ and $\F_{i-1} \neq \uA{\R}_O(\F_{i-1}) $, then $\F_i = \uA{\R}_O(\F_{i-1})$.
\item Otherwise, $\F_{i} = \eA{\R}_O(\F_{i-1})$.
\end{itemize}
The \emph{restricted chase sequence} of \P is defined analogously.
\end{definition}

For the sake of brevity, we frequently denote the oblivious (resp., restricted) chase sequence of a program $\P$ with $\oCS{\P}{1}, \oCS{\P}{2}, \ldots$ (resp., $\rCS{\P}{1}, \rCS{\P}{2}, \ldots$)

\begin{definition}
\label{definition:chase}
Let \P be some program and let \R be some set of rules.
Then, the \emph{oblivious chase} of \P is the set $\oChase{\P} = \bigcup_{i \in \Naturals} \oCS{\P}{i}$.
The \emph{restricted chase} of \P, written $\rChase{\P}$, is defined analogously.

The oblivious (resp., restricted) chase of \P \emph{terminates} if and only if there is some $i$ such that, for all $j \geq i$, $\oCS{\P}{i} = \oCS{\P}{j}$.
Furthermore, the oblivious (resp., restricted) chase of a set of rules \R \emph{terminates} if the oblivious (resp., restricted) chase of every program of the form \RI terminates.
\end{definition}

Our definition of the chase sequence ensures that rules which do not contain existentially quantified variables are always applied with a higher priority than rules that do.
Note that, by postponing the application of rules with existential variables, we may prevent them from introducing further consequences.

The (restricted or oblivious) chase of a program can be employed to solve CQ entailment \cite{CaliGK08}.
I.e., a program \P entails a query \query, written $\P \models \query$, if and only if either $\oChase{\P} \models \exists y \bot(y)$ or $\oChase{\P} \models \query$ (resp., $\rChase{\P} \models \exists y \bot(y)$ or $\rChase{\P} \models \query$).
Thus, we may also use the chase to solve CQ entailment over ontologies: An ontology \O entails a query \query if and only if $\oChase{\Po} \models \exists y \bot(y)$ or $\oChase{\Po} \models \query$ (resp., $\rChase{\Po} \models \exists y \bot(y)$ or $\rChase{\Po} \models \query$).

For readability purposes, 
we say that the oblivious (resp. restricted) chase of some ontology \O \emph{terminates} if and only if the oblivious (resp. restricted) chase of \Po terminates.
The oblivious (resp. restricted) chase of some TBox \T \emph{terminates} if and only if if the oblivious (resp. restricted) chase of $\Rt$ terminates.

\begin{figure}[t]
\begin{align*}
\T = \{		& \Film \sqsubseteq \exists \isProducedBy.\Producer, \Producer \sqsubseteq \exists \produces.\Film, \\
			&\isProducedBy^- \sqsubseteq \produces, \produces^- \sqsubseteq \isProducedBy \}\\
\O = \langle	& \T, \{\Film(\AI)\}\rangle \\
\Rt = \{		& \rho = \Film(x) \to \exists y [\isProducedBy(x, y) \wedge \Producer(y)], \\
			& \auxRule = \Producer(x) \to \exists y [\produces(x, y) \wedge \Film(y)], \\
			& \isProducedBy(y, x) \to \produces(x, y), \produces(y, x) \to \isProducedBy(x, y) \}\\
\Po = \langle	& \Rt, \{\Film(\AI)\} \rangle \\
\Po_R^1 = \{	& \Film(\AI), \isProducedBy(\AI, \sF{y}{\rho}(\AI)), \Producer(\sF{y}{\rho}(\AI))\} \\
\Po_R^2 = \{	& \produces(\sF{y}{\rho}(\AI), \AI)\} \cup \Po_O^1 \\
\rChase{\Po} =\phantom{\{}	&\Po_O^2 \\
\oChase{\Po} =  \phantom{\{}	&\rChase{\Po} \cup \{\produces(\sF{y}{\rho}(\AI), \sF{y}{\auxRule}(\sF{y}{\rho}(\AI))), \Film(\sF{y}{\auxRule}(\sF{y}{\rho}(\AI))), \ldots\}
\end{align*}
\caption{Ontology $\O = \TA$, program \Po and the chase of \Po.}
\label{figure:example}
\end{figure}

As expected, the restricted chase has a better behavior than the oblivious chase; i.e., in some cases, the former might terminate when the latter does not: 
 
\begin{example}
\label{mainExample}
Let \O = \TA be as in Figure~\ref{figure:example}.
The figure depicts also the computation of the oblivious chase and that of the restricted chase of \Po.
In this case, $\rChase{\Po}$ terminates whereas $\oChase{\Po}$ does not. 
\end{example}

\section{Model Faithful Acyclicity}
\label{section:acyclicityNotions}

In this section we briefly describe Model Faithful Acyclicity (\MFA) \cite{CG+13:acyclicity}, one of the most general acyclicity conditions for sets of rules. 
\MFA guarantees the termination of the oblivious chase of a program by imposing that no cyclic term occurs in the chase. 
Note that, a condition such as \MFA can be applied to check whether a TBox \T is acyclic; i.e., \T is \MFA if and only if \Rt is \MFA.


When one is interested in checking the termination of the oblivious chase with respect to every possible instance, it is enough to check termination with respect to a special instance, the \emph{critical instance}  \cite{DBLP:conf/pods/Marnette09}.
The critical instance is the minimal set which contains all possible atoms that can be formed using the relational symbols which occur in TGDs and the special constant  $\star$.
Such a strategy is used by \MFA to guarantee termination of a set of rules.


While the actual definition of \MFA does not preclude the existence of EGDs, equality is assumed to be axiomatized, and thus it is treated as a regular predicate (EGDs are de facto TGDs).
To reflect such treatment we will use the special predicate $\eP$ to denote equality.
However, as the following example shows, the presence of equality in a set of TGDs frequently makes the \MFA membership test fail.

\begin{example}
\label{ex:mfaeq}
Let $\Sigma$ be the following set of rules and let $\Sigma'$ be the set of rules that result from axiomatizing the equality predicate as usual (see Section 2.1 of \cite{CG+13:acyclicity}).
Furthermore, let $\cI{\Sigma'}$ be the critical instance of $\Sigma'$.
\begin{align*}
\Sigma = \{&A(x) \wedge B(x) \to \exists y [R(x,y) \wedge B(y)], R(z, x_1) \wedge R(z,x_2) \to \eP(x_1, x_2)\} \\
\equalitySet = \{&\top(x) \to \eP(x, x), \eP(x, y) \to \eP(y, x), \eP( x, z) \wedge \eP(z, y) \to \eP(x, y)\} \\
\Sigma' = \{& A(x) \wedge \eP(x, y) \to A(y), R(x, y) \wedge \eP(x, z) \to R(z, y), \\
& R(x, y) \wedge \eP(y, z) \to R(x, z)\} \cup \Sigma \cup \equalitySet \\
\cI{\Sigma'} = \{&A(\star), R(\star, \star), \eP(\star,\star) \}
\end{align*}

The oblivious  chase of $(\Sigma', \cI{\Sigma'})$ does not terminate.
\begin{align*}
 \oCS{(\Sigma', \cI{\Sigma'} )}{1} = \{&R(\star, f(\star)), B(f(\star)), \eP(\star, f(\star)) \}  \cup \cI{\Sigma'}\\
\oCS{(\Sigma', \cI{\Sigma'} )}{2} = \{&A(f(\star)), R(f(\star), f(f(\star))), B(f(f(\star))), \ldots \} \\
\ldots\ldots\ldots\ldots&\ldots\ldots\ldots\ldots
\end{align*}
\end{example}

To avoid this situation, the use of \emph{singularization} \cite{DBLP:conf/pods/Marnette09}, a somewhat ``less-harmful'' axiomatization of equality, is proposed in \cite{CG+13:acyclicity}.

\begin{definition}
A \emph{singularization of a rule $\rho$} is the rule $\rho'$ that results from performing the following transformation for every variable $v$ in the body of $\rho$:
\begin{itemize}
\item Rename each occurrence of $v$ using different fresh variables $v_1, \ldots, v_n$,
\item pick some $j = 1, \ldots, n$ and add the atoms $\eP(v_1, v_j), \ldots, \eP(v_n, v_j)$ to the body of $\rho$ and
\item replace any occurrence of $v$ in the head of $\rho$ with $v_j$.
\end{itemize}

Let $\Sigma$ be a set of TGDs and let $\textsf{Eq}$ be the set from Example \ref{ex:mfaeq}.
A \emph{singularization} of $\Sigma$ is a set of TGDs $\Sigma'$ which contains $\textsf{Eq}$ and exactly one singularization of every $\rho \in \Sigma$.
Let $\textit{Sing}(\Sigma)$ be the set of all possible singularizations of $\Sigma$.
\end{definition}

\begin{example} 
Rule $A(x) \wedge B(x) \to \exists y [R(x,y) \wedge B(y)]$ from Example \ref{ex:mfaeq} admits two possible  singularizations: 
(i) $A(x_1) \wedge B(x_2) \wedge \eP(x_2, x_1) \to \exists y [R(x_1,y) \wedge B(y)]$ and (ii) $A(x_1) \wedge B(x_2) \wedge \eP(x_1, x_2) \to \exists y [R(x_2, y) \wedge B(y)]$.
\end{example}

Note that, for any $\Sigma' \in \textit{Sing}(\Sigma)$, if $\Sigma'$ is MFA, then the oblivious chase of $\Sigma'$ can be used to answer queries on $\Sigma$ \cite{CG+13:acyclicity}.
The use of singularization along with \MFA gives rise to the following acyclicity notions.

\begin{definition}
For a set of TGDs $\Sigma$, if there is some $\Sigma' \in \textit{Sing}(\Sigma)$ which is \MFA, then $\Sigma$ is said to be $\MFA^\exists$.
If every $\Sigma' \in \textit{Sing}(\Sigma)$ is \MFA, then $\Sigma$ is \MFAF.
\end{definition}

To some extent, the use of singularization solves the problems with equality: One can check that $\Sigma$ in Example \ref{ex:mfaeq} is \MFAE, but not \MFAF.
Nevertheless, due to the high number of possible singularizations, it is frequently not feasible to check \MFAE or \MFAF membership.
A simpler alternative is to check whether $\bigcup_{\Sigma' \in \textit{Sing}(\Sigma)} \Sigma'$ is $\MFA$.
If that is the case, then $\Sigma$ is said to be \MFAC.
Note that in the case of \hsriq TBoxes, $\vert \bigcup_{\Sigma' \in \textit{Sing}(\Sigma)} \Sigma' \vert$ is actually polynomial in $\vert \Sigma \vert$ and, as such, \MFAC is more feasible to check. 
Thus, we will use \MFAC as a baseline for the evaluation of the new acyclicity condition \RCA{n}, which is introduced in the next section.

\section{Restricted Chase Acyclicity}
\label{section:RCA}

While \MFA is quite a general acyclicity condition, it has two main drawbacks: 

\begin{enumerate}
\item It only considers the oblivious chase, which as we have seen in Example \ref{mainExample}, might not terminate (even though the restricted chase does!), and 
\item its treatment of equality via singularization is cumbersome and inefficient in practice.
Not only \MFAE and \MFAF are difficult to check, but even after a set of TGDs are established to belong to some $\MFA$ subclass, one has to employ a singularized program for reasoning purposes. 
\end{enumerate}

In this section, we present \RCA{n}, an acyclicity notion with neither of these drawbacks: \RCA{n} verifies termination of the restricted chase of a TBox and does not require the use of cumbersome axiomatizations of the equality predicate.
Furthermore, unlike \MFA, \RCA{n} allows for the presence of cyclic terms in the chase up to a given depth $n$. 

Since we are primarily interested in termination of the restricted chase of a \hsriq TBox, one might wonder why we do not simply check for termination of the restricted chase for such a TBox with respect to the critical instance, as it is done in the previous section with the oblivious chase.
Unfortunately, this is not possible: The restricted chase of any set of existential rules always terminates with respect to the critical instance.
Thus, we have to devise more sophisticated techniques to check the termination of the restricted chase.  
We start by introducing the notion of an overchase for a TBox.

\begin{definition}
\label{definition:overchase}
A set of facts \V is an \emph{overchase} for some TBox \T if and only if, for every \OTA, $\rChase{\Po}_\star \subseteq \V$.
\end{definition}

Given some TBox \T, an overchase for \T may be intuitively regarded as an over-approximation of the restricted chase of \T.

\begin{lemma}
\label{lemma:overchaseTermination}
If there exists a finite overchase for a TBox, then the restricted chase of such TBox terminates.
\end{lemma}

Thus, to determine whether the chase of a TBox \T terminates, we introduce a procedure to compute an overchase for \T and a means to check its termination.
We proceed with some preliminary notions and notation.

\begin{definition}
\label{definition:restrictedProgram}
Let \T be some TBox and $t$ a term.
Let $\I(t)$ be the set of facts defined as follows: If $t$ is of the form $\sF{y}{\rho}(s)$ where $\rho = A(x) \to \exists y [R(x, y) \wedge B(y)]$, then $\I(t) = \{A(s), R(s, t), B(t)\} \cup \I(s)$; otherwise, $\I(t) = \emptyset$.
Furthermore, we introduce the program $\RT{\T}{t} = \langle \uA{\Rt} \cup \eR{\Rt}, \I(t) \rangle$.
\end{definition}

Intuitively, the restricted chase of the program \RT{\T}{t} can be regarded as some kind of under-approximation of the facts that must occur in the chase of every program of the form \Pta where $t$ occurs.
I.e., if $t$ occurs in the restricted chase sequence of any program \Pta, then the facts in the restricted chase of \RT{\T}{t} must also occur (up to renaming) in the chase sequence of such program.
Furthermore, due to the special priority of application of the rules during the computation of the chase, the facts in the restricted chase of \RT{\T}{t} must occur in the restricted chase sequence of every program of the form \Pta before any successors of $t$ are introduced.

\begin{example}
Let \O, $\rho$ and \auxRule be the ontology and rules from Example \ref{mainExample}.
Then, by Definition \ref{definition:restrictedProgram}:
\begin{align*}
\I(\sF{y}{\rho}(\AI)) = \{&\Film(\AI), \isProducedBy(\AI, \sF{y}{\rho}(\AI)), \Producer(\sF{y}{\rho}(\AI))\} \text{ and }\\
\rChase{\RT{\T}{\sF{y}{\rho}(\AI)}} = \{&\produces(\sF{y}{\rho}(\AI), \AI)\} \cup \I(\sF{y}{\rho}(\AI)).
\end{align*}

All the facts in the restricted chase of \RT{\T}{t} occur in the restricted chase sequence of \Po before any successors of term $\sF{y}{\rho}(\AI)$ are introduced.
This is because the rule $\isProducedBy(y, x) \to \produces(x, y)$ is applied with a higher priority than the rule $\auxRule = \Producer(x) \to \exists y [\produces(x, y) \wedge \Film(y)]$.
\end{example}

Given a TBox \T and some term of the form $\sF{y}{\rho}(t)$, we can in some cases conclude that such a term may never occur during the computation of the restricted chase of every program of the form \Pta by carefully inspecting the facts in the set \RT{\T}{t}.

\begin{definition}
\label{definition:restrictedTerm}
Let \T be a TBox and $t$ a term of the form $\sF{y}{\rho}(s)$ where $\rho = A(x) \to \exists y [R(x, y) \wedge B(y)]$.
We say that a term $t$ is \emph{restricted with respect to \T} if and only if there is some term $u$ with $\{R([s], u), B(u)\} \subseteq \rChase{\RT{\T}{s}}$ where $[s] = [v]$, if $s$ is replaced by $v$ during the computation of the restricted chase sequence; and $[s] = s$, otherwise.
\end{definition}

We often simply say that a term is ``restricted'', instead of ``restricted with respect to \T,'' if the TBox \T is clear from the context.

\begin{lemma}
\label{lemma:restrictedTerm}
Let \T be a TBox and $t$ a restricted term.
Then, for every possible \OTA, $t \notin \rChase{\Po}$.
\end{lemma}

\begin{proof}
(Sketch)
Let $t$ be a term of the form $\sF{y}{\rho}(s)$ where $\rho = A(x) \to \exists y(R(x, y) \wedge B(y))$.
We can verify that, if $t$ occurs during the computation of the chase sequence, then every fact \rChase{\RT{\T}{s}} will also be included in such chase sequence before any new terms are introduced.
Thus, if $t$ is indeed restricted, there must be some $u$ with $R([s], u)$ and $B(u)$ occurring in the chase sequence.
Therefore, by the definition of the chase, the term $t$ may never be derived.
\end{proof}

\begin{example}
Let \T, $\rho$ and \auxRule be the TBox and rules from Example \ref{mainExample}.
We proceed to show that the term $\sF{y}{\rho}(\sF{y}{\auxRule}(\AI))$ is restricted.
First, we compute the restricted chase of $\RT{\T}{\sF{y}{\auxRule}(\AI)}$.
\begin{align*}
\rChase{\RT{\T}{\sF{y}{\auxRule}(\AI)}} = \{&\Producer(\AI), \produces(\AI, \sF{y}{\auxRule}(\AI)), \\
&\Film(\sF{y}{\auxRule}(\AI)), \isProducedBy(\sF{y}{\auxRule}(\AI), \AI)\}
\end{align*}
Note that $\{\isProducedBy(\sF{y}{\auxRule}(\AI), \AI), \Producer(\AI)\} \subseteq \rChase{\RT{\T}{\sF{y}{\auxRule}(\AI)}}$.
Thus, $\sF{y}{\rho}(\sF{y}{\auxRule}(\AI))$ is restricted with respect to \T and, by Lemma \ref{lemma:restrictedTerm}, it may not occur in the restricted chase of a program of the form \Pta.
Furthermore, by Definition \ref{definition:restrictedTerm}, if $\sF{y}{\rho}(\sF{y}{\auxRule}(\AI))$ is restricted, then every term of the form $\sF{y}{\rho}(\sF{y}{\auxRule}(c))$, where $c$ is a constant, is also restricted.
\end{example}

With Definition \ref{definition:restrictedTerm} and Lemma \ref{lemma:restrictedTerm} in place, we proceed with the definition of a procedure to construct an overchase for some given TBox \T.

\begin{definition}
\label{definition:VOverchase}
Let \T be a TBox.
We define \overchase{\T} as the set initially containing every fact in $\cI{\Rt}$ which is then expanded by repeatedly applying the rules in Figure \ref{figure:overchaseRules} (in non-deterministic order).
\end{definition}

\begin{figure*}[t]
\centering
\begin{tabu} to \linewidth {X[1,l] X[0.8,l] X[9,l]}
\uRule	&if		& there is some TGD of the form  $\rho = \body(\vx, \vy) \to \head(\vx) \in \Rt$ \\
		&then	& $\overchase{\T} \rightarrow \rho_R(\overchase{\T}) \cup \overchase{\T}$ \\ \hline
\eRule	&if		& there is some TGD of the form  $\rho = A(x) \to \exists y [R(x, y) \wedge B(y)] \in\Rt$ and there exists some substitution \subs such that (i) $A(x) \subs \subseteq \overchase{\T}$ and  (ii) $\sF{y}{\rho}(x)\subs$ is not restricted with respect to \T \\		
		&then	& $\overchase{\T} \rightarrow \{R(x, \sF{y}{\rho}(x)), B(\sF{y}{\rho}(x))\}\subs \cup \overchase{\T}$ \\ \hline
\aRule	&if		& there is some EGD $\body(\vx, \vy) \to  x \approx y \in \Rt$ and there exists some substitution \subs such that $\body(\vx, \vy) \subs \subseteq \overchase{\T}$ \\
		&then	& $\overchase{\T} \rightarrow \{\eP(x, y), \eP(y, x)\}\subs \cup \overchase{\T}$ \\ \hline
\eqRule	&if		& there are some terms $t$, $u$ and $u_i$ where $i = 1, \ldots, n$ and some predicate $p$ such that (i) $p \neq \eP$, (ii) $\{\eP(t, u), p(u_1, \ldots, u_n)\} \subseteq \overchase{\T}$, (iii) $\depth(t) \leq \depth(u)$ and (iv) $u = u_j$ for some $j = 1, \ldots, n$\\
		&then	& $\overchase{\T} \rightarrow \{p(u_1, \ldots, u_n)\}[u/t] \cup \overchase{\T}$
\end{tabu}
\caption{Expansion rules for the construction of \overchase{\T}.}
\label{figure:overchaseRules}
\end{figure*}

\begin{lemma}
\label{lemma:VOverchase}
The set \overchase{\T} is an overchase of the TBox \T.
\end{lemma}
\begin{proof}
(Sketch)
The lemma can be proven via induction on chase sequence of any ontology of the form \OTA.
Note that, $\rCS{\O}{0} \subseteq \overchase{\T}$ by the definition of \overchase{\T}.
It can be verified that, for every possible derivation of a set of facts during the computation of the chase of \O, such facts will always be contained in \overchase{\T}.
\end{proof}

\begin{corollary}
\label{corollary1}
The restricted chase of some TBox \T terminates if \overchase{\T} is finite.
\end{corollary}

\begin{example}
Let \T be the TBox from Example \ref{mainExample}.
Then \overchase{\T} is as follows.
\begin{align*}
\overchase{\T} = \{	&\Film(\star), \isProducedBy(\star), \Producer(\star), \produces(\star, \star), \\
				&\isProducedBy(\star, \sF{y}{\rho}(\star)), \Producer(\sF{y}{\rho}(\star)), \produces(\star, \sF{y}{\auxRule}(\star)), \Producer(\sF{y}{\auxRule}(\star))\}
\end{align*}

Note that terms $\sF{y}{\rho}(\sF{y}{\auxRule}(\star))$ and $\sF{y}{\auxRule}(\sF{y}{\rho}(\star))$ are restricted and thus, they are not included in \overchase{\T}.
Since \overchase{\T} is finite, we can conclude termination of the restricted chase of the TBox \T.
\end{example}

In the previous example, we were able to ascertain termination of the restricted chase of \T after verifying that the set \overchase{\T} is finite.
A sufficient condition for finiteness of \overchase{\T} is to only allow cyclic terms up to a certain depth in this set.
We use such condition to formally define \RCA{n}.

\begin{definition}
\label{definition:RCA}
A TBox \T \emph{is \RCA{n}} if and only if there are no $n$-cyclic terms in \overchase{\T}.
An ontology \TA is \RCA{n} if and only if \T is \RCA{n}.
\end{definition}

\begin{theorem}
\label{theorem:RCA}
If a TBox \T is \RCA{n} then the restricted chase of \T terminates.
\end{theorem}

We proceed with several results regarding the complexity of deciding \RCA{n} membership and reasoning over \RCA{n} ontologies.

\begin{theorem}
\label{lemma:RCAComplexityExp}
Deciding whether some TBox \T is \RCA{n} is in \ExpTime.
\end{theorem}


\begin{theorem}
\label{lemma:cqComplexity}
Let $\O = \TA$ be some \RCA{n} ontology and \query a query.
Then, checking whether $\O \models \query$ is \ExpTimeC.
\end{theorem}

To close the section, we present several results in which we theoretically compare the generality of \RCA{n} to \MFAC.

\begin{theorem}
\MFAC does not cover \RCA{1}.
\end{theorem}

\begin{proof}
The TBox \T from Example \ref{mainExample} is \RCA{1}  but not \MFAC.
\end{proof}

\begin{theorem}
\label{lemma:MFACImpliesRCA}
If \T is \MFAC then \T is \RCA{n} for every $n > \vert \eA{\T} \vert$ where \eA{\T} is the set of all existential axioms in \T.
\end{theorem}

\section{Evaluation}
\label{section:evaluation}

\subsection{An Empirical Comparison of \RCA{n} and \MFAC}

In this section we include an empirical comparison of the generality of \RCA{n} and \MFAC.
For our experiments, we use the TBoxes of the ontologies in the OWL Reasoner Evaluation workshop (ORE, \url{https://www.w3.org/community/owled/ore-2015-workshop/}) and Ontology Design Patterns (ODP, \url{http://www.ontologydesignpatterns.org}) datasets.
The former is a large repository used in the ORE competition containing a large corpus of ontologies.
The latter contains a wide range of smaller ontologies that capture design patterns commonly used in ontology modeling.
The ORE dataset is rather large, and thus we restrict our experiments to the 294 ontologies with the smallest number of existential axioms, while skipping the 77 ontologies with the largest number of existential axioms. The number of such axioms contained in an ontology is a useful metric to predict the ``hardness'' of acyclicity membership tests; i.e. running these experiments would be very time-intensive, while our results, reported below, already indicate that for such very hard TBoxes \MFAC and \RCA{n} will likely not differ much (while they differ significantly for ontologies with a lower count of existential axioms).

Only \hsriq TBoxes which cannot be expressed in any of the OWL 2 profiles were considered in our experiments.
This is because all OWL 2 RL TBoxes are acyclic (with respect to every applicable acyclicity notion known to us), and there already exist effective algorithms and efficient implementations that solve CQ answering over OWL 2 EL and OWL 2 QL ontologies \cite{DBLP:conf/kr/KontchakovLTWZ10,DBLP:conf/aaai/StefanoniMH13,DBLP:journals/jair/StefanoniMKR14} (albeit, if these do not include complex roles).

The results from our experiments are summarized in Figure \ref{figure:RCAvsMFA}.
The evaluated TBoxes are sorted into brackets depending on the number of existential axioms they contain.
For each bracket we provide the average number of axioms in the ontologies (``Avg. Size''), the number of ontologies (``Count''), and, for every condition ``X'' considered, the percentage of ``X acyclic'' ontologies

\RCA{2} and \RCA{3} turned out to be indistinguishable with respect to the TBoxes considered and thus, we limit our evaluation to \RCA{n} with $n \leq 3$.
Our tests reveal that \RCA{2} is significantly more general than \MFAC, particularly when it comes to TBoxes with a low count of existential axioms. However note that reasoning over ontologies with few (existential) axioms is in general not trivial: All of the ontologies considered in our materialization tests (see Figure \ref{figure:reasoningEvaluation}) contain less than 20 existential axioms. 
For TBoxes containing from 1 to 10 existential axioms in the ORE dataset, more than half of the ontologies which are not \MFAC are \RCA{2}.
Furthermore, the 4 ontologies in the ODP dataset which are not \MFAC are \RCA{2}.
Interestingly, in both repositories we could not find any ontology that is \MFAC but not \RCA{1}.
Thus, with respect to the TBoxes in our corpus, \RCA{1} already proves to be more general than \MFAC.

\begin{figure*}[t]
\centering
\begin{tabu} to \linewidth {X[0.3,l] X[1,l] X[1,l] X[1,l] X[1,l] X[1,l] X[1,l] X[1,l] X[0.3,l]}
&ORE \\ \cline{2-8}
&$\exists$-Axioms	& Avg. Size 	& Count	& \MFAC	& \RCA{1}	& \RCA{2}	& \RCA{3}	&\\ \cline{2-8}
&1-5				& 175		& 70		& 70.0	& 87.1	& 92.9	& 92.9	&\\
&6-10			& 219		& 48		& 58.3	& 83.3	& 83.3	& 83.3	&\\
&11-25			& 916		& 54		& 83.3	& 85.2	& 91	& 91	&\\
&26-100			& 521		& 42		& 54.8	& 59.5	& 61.9	& 61.9	&\\
&101-500			& 1290		& 42		& 26.2	& 26.2	& 28.6	& 28.6	&\\
&501-1922		& 5052		& 38		& 60.5	& 60.5	& 60.5	& 60.5	&\\ \cline{2-8}
&1-1922			&1362		& 294	& 60.9	& 70.1	& 73.1	& 73.1	&\\ \\
&ODP \\ \cline{2-8}
&$\exists$-Axioms	& Size		& Total	& \MFAC	& \RCA{1}	& \RCA{2}	& \RCA{3}	&\\ \cline{2-8}
&1-12			& 39			& 18		& 73.7	& 100.0	& 100.0	& 100.0
\end{tabu}
\caption{Results for the ORE and ODP Repositories.}
\label{figure:RCAvsMFA}
\end{figure*} 

In total, we looked at 312 ontologies, $62\%$ and $75\%$ of which are \MFAC and \RCA{2}, respectively.
To gauge the significance of this improvement, we roughly compare these numbers with the results presented in \cite{CG+13:acyclicity}.
In that paper, the authors consider a total of 336 ontologies, of which $49\%$, $58\%$ and $68\%$ are \emph{weakly acyclic} \cite{DBLP:journals/tcs/FaginKMP05}, \emph{jointly acyclic} \cite{DBLP:conf/ijcai/KrotzschR11} and \MFAC, respectively.
Even though the comparison is not over the same TBoxes, we verify that the improvement in generality of our notion is in line with previous iterations of related work.

\subsection{A Materialization Based Reasoner}
\label{section:materialization}

We now report on an implementation of the restricted chase as defined in Section~\ref{section:reasoningChase}.
Moreover, we also present an implementation of the oblivious chase with singularization, i.e., the chase as it must be used if we employ \MFAC (see Section~\ref{section:acyclicityNotions}).
We use the datalog engine RDF\-Ox \cite{DBLP:conf/semweb/NenovPMHWB15} in both implementations.

We evaluate the performance of our chase based implementations against Konclude \cite{DBLP:journals/ws/SteigmillerLG14}, a very efficient OWL DL reasoner, and PAGOdA \cite{DBLP:journals/jair/ZhouGNKH15}, a hybrid approach to query answering over ontologies.
PAGOdA combines a datalog reasoner with a fully-fledged OWL 2 reasoner in order to provide scalable 'pay-as-you-go' performance and is, to the best of our knowledge, the only other implementation that may solve CQ answering over \hsriq ontologies with completeness guarantees, albeit only in some cases.
Nevertheless, PAGOdA was able to solve all the queries (that is, all of which for which it did not time-out or run out of memory) in this evaluation in a sound and complete manner.



We consider two real-world ontologies in our experiments, Reactome and Uniprot, and two standard benchmarks, LUBM and UOBM, all of which contain a large amount of ABox axioms.
Axioms in these ontologies which are not expressible in \hsriq were pruned.
Furthermore, one extra axiom had to be removed from Uniprot for it to be both \MFAC and \RCA{1} acyclic.

\begin{figure*}[t]
\centering
\begin{tabular}{| r | r | rrrr  | r | rrrr | r | rrrr | r |}
\hline
Triples~	& \multicolumn{5}{c|}{Restricted}	& \multicolumn{5}{c|}{Oblivious}		& \multicolumn{5}{c|}{PAGOdA}		& \multicolumn{1}{c|}{Konc.}	\\ 
\cline{2-17}
Count~	& \multicolumn{1}{c|}{C} &\multicolumn{4}{c|}{Q1-Q4} & \multicolumn{1}{c|}{C} &\multicolumn{4}{c|}{Q1-Q4} & \multicolumn{1}{c|}{P} &\multicolumn{4}{c|}{Q1-Q4} & \multicolumn{1}{c|}{R}	\\ \hline
~$2.8M$	~&~ 10	~&~ 0 & 0 & 0 & 0 		~&~ 45	~&~ 0 & 0 & TO & 0	~&~ 89	~&~ OM & 4 & 1 & 0		~&~ 75 ~\\
~$5.1M$	~&~ 21	~&~ 0 & 0 & 0 & 0		~&~ 138	~&~ 0 & 0 &TO & 3	~&~ 147	~&~ OM & 1 & 2 & 0		~&~ 214 ~\\
~$6.7M$	~&~ 28	~&~ 0 & 0 & 0 & 0		~&~ 1029	~&~ 2 & 0 & TO & 0	~&~ 203	~&~ OM & 2 & 3 & 1		~&~ 506 ~\\
~$8.1M$	~&~ 36	~&~ 37 & 0 & 0 & 0		~&~ TO	~&~ - & - & - & -	~&~ 263	~&~ OM & 2 & 2 & 6		~&~ 1347 ~\\ \hline
~$9.0M$	~&~ 37	~&~ 0 & 0 & 0 & 0		~&~ OM	~&~ - & - & - & - 	~&~ 113	~&~ 1 & 1 & 1 & 1		~&~ 198 ~\\	
~$17.8M$	~&~ 72	~&~ 0 & 0 & 0 & 0		~&~ OM	~&~ - & - & - & - 	~&~ 232	~&~ 2 & 2 & 3 & 3		~&~ 987 ~\\
~$26.2M$	~&~ 107	~&~ 0 & 0 & 0 & 0		~&~ OM	~&~ - & - & - & - 	~&~ 378	~&~ 4 & 10 & 12 & 5		~&~ 3491 ~\\
~$33.9M$	~&~ 141	~&~ 0 & 1 & 0 & 0		~&~ OM	~&~ - & - & - & - 	~&~ 521	~&~ 6 & 21 & 21 & 12	~&~ TO ~\\ \hline
~$2.8M$	~&~ 8	~&~ 0 & 0 & 0 & 1		~&~ 70	~&~ 0 & 0 & 0 & 74	~&~ 51	~&~ OM & 0 & 0 & 0		~&~ 51 ~\\
~$5.7M$	~&~ 16	~&~ 0 & 0 & 0 & 2		~&~ 158	~&~ 1 & 1 & 1 & 154	~&~ 99	~&~ OM & 1 & 1 & 0		~&~ 118 ~\\
~$8.4M$	~&~ 26	~&~ 0 & 0 & 0 & 3		~&~ 242	~&~ 1 & 1 & 2 &186	~&~ 142	~&~ OM & 2 & 1 & 1		~&~ 220 ~\\
~$11.4M$	~&~ 37	~&~ 1 & 0 & 0 & 5		~&~ 341	~&~ 2 & 2 & 3 & 311	~&~ 197	~&~ OM & 3 & 1 & 1		~&~ 315 ~\\ \hline
~$2.2M$	~&~ 11	~&~ 0 & 0 & 0 & 0		~&~ 56	~&~ 0 & 0 & 0 & 1	~&~ 61	~&~ 28 & 0 & TO & 1	~&~ 53 ~\\
~$4.5M$	~&~ 27	~&~ 2 & 0 & 0 & 0		~&~ 133	~&~ 0 & 0 & 1 & 2	~&~ 121	~&~ 60 & 0 & TO & 2	~&~ 125 ~\\
~$6.6M$	~&~ 42	~&~ 3 & 1 & 1 & 0		~&~ 216	~&~ 1 & 1 & 2 & 3	~&~ 186	~&~ TO & 0  & TO & 5 	~&~ 292 ~\\
~$8.9M$	~&~ 58	~&~ 5 & 1 & 2 & 1		~&~ 310	~&~ 1 & 2 & 4 & 6	~&~ 260	~&~ TO & 0 & TO & 5	~&~ 644 ~\\ \hline
\end{tabular}
\caption{Results for Reactome, Uniprot, LUBM and UOBM (sorted from top to bottom in the above table).
} 
\label{figure:reasoningEvaluation}
\end{figure*}

The results from our experiments are summarized in Figure \ref{figure:reasoningEvaluation}.
For each ontology, we consider four samples of the original ABox.
The number of triples contained in each one of these is indicated at the beginning of each row, under the column ``Triples Count''.
As previously mentioned, we consider four different implementations: These include the two aforementioned variants of the chase (``Restricted'' and ``Oblivious''), PAGOdA (``PAGOdA'') and Konclude (``Konc.'').
For both chase based implementations, we check the time it takes to  compute the chase (``C'') and then the time to solve each of the four queries crafted for each ontology (``Q1-Q4'').
In a similar manner, we list the time PAGOdA takes to preprocess each ontology (``P'') plus the time it takes to answer the queries (``Q1-Q4'').
Finally, we list the time Konclude takes to solve realization; i.e., the task of computing every fact of the form $A(a)$ entailed by an ontology (note that Konclude cannot solve arbitrary CQ answering).
Time-outs, indicated with ``TO,'' were set at 1 hour for materialization and 5 minutes for queries.
We make use of the acronym ``OM'' to indicate that an out-of-memory error occurred.
Sometimes, a time-out or an out of memory error prevents us from answering the queries: Such a situation is indicated with ``-.''
All experiments were performed on a MacBook Pro with 8GB of RAM and a 2.4 GHz Intel Core i5 processor.

For each ontology, we consider four different queries which are listed in the App. Section \ref{appendix:queries} included in the extended technical report.
A summarized description of these queries, in which we ignore unary predicates, can be found in Figure \ref{figure:queries}.
For every ontology, the query Q1 is of the form $\exists x, y, z R(x, y) \wedge R(z, y)$ where $R$ is an existentially quantified role occurring in the TBox.
It appears that PAGOdA has trouble with this kind of query, whereas the chase based implementations efficiently solve it in all but one case.
This is probably due to the design of the hybrid reasoner which considers under and over approximations to provide complete answers to CQ: It appears that queries as the one previously considered find a large number of matches in the upper bound which slows down the performance of this reasoner.
Queries Q2, and Q3 and Q4 are acyclic and cyclic, respectively (a query is acyclic if the shape of its body is acyclic).
Even though it is well-known that answering acyclic CQs can be reduced to satisfiability \cite{DBLP:conf/esws/MartinezH12}, we included such a type of query in our evaluation in an attempt to verify whether solving acyclic queries is simpler than cyclic queries (this is indeed the case theoretically).
Nevertheless, our experiments do not reveal any significant differences.

\begin{figure}[t]
\scriptsize
\begin{align*}
\textsf{q}_1(w, y) : \text{ } & \textsf{pE}(w, z), \textsf{pE}(y, z) 								& \textsf{q}_1(x, y) : \text{ } &\textsf{cC}(x, z), \textsf{cC}(y, z) \\
\textsf{q}_2(x, z) : \text{ } & \textsf{mPE}(z, w), \textsf{mPE}(z, w), \textsf{p}(y, z), \textsf{pC}(x, y)	& \textsf{q}_2(x) : \text{ } &\textsf{tF}(w, x), \textsf{lO}(x, y), \textsf{d}(x, z) \\
\textsf{q}_3(x, z) : \text{ } & \textsf{fL}(x, w), \textsf{fL}(x, y), \textsf{sIB}(w, z), \textsf{sIB}(y, z)		& \textsf{q}_3(x) : \text{ } &\textsf{tF}(w, y), \textsf{tF}(w, x), \textsf{d}(y, z), \textsf{d}(x, z) \\
\textsf{q}_4(x, z) : \text{ } &\textsf{p}(w, z), \textsf{p}(y, z), \textsf{pC}(x, w), \textsf{pC}(x, y) 		& \textsf{q}_4(x) : \text{ } &\textsf{lI}(x, w), \textsf{cC}(w, z), \textsf{lI}(x, y), \textsf{cC}(y, z) \\ \\
\textsf{q}_1(x, z) : \text{ } & \textsf{wF}(x, y), \textsf{wF}(z, y), \textsf{pA}(x, z)					& \textsf{q}_1(x, y) : \text{ } &\textsf{tC}(x, z), \textsf{tC}(y, z) \\
\textsf{q}_2(x) : \text{ } & \textsf{a}(x, y), \textsf{tO}(y, z), \textsf{mO}(y, w)						& \textsf{q}_2(x) : \text{ } &\textsf{tAO}(x, y), \textsf{pA}(z,x), \textsf{tC}(w, y), \textsf{wF}(x, v)\} \\
\textsf{q}_3(x, z) : \text{ } & \textsf{tO}(y, z), \textsf{a}(x, y), \textsf{tC}(x, z)						& \textsf{q}_3(x, y) : \text{ } & \textsf{iFO}(x, y), \textsf{l}(x, z) \\
\textsf{q}_4(x) : \text{ } &\textsf{pA}(x, z), \textsf{pA}(x, y), \textsf{a}(z, y),						& \textsf{q}_4(x, y) : \text{ } & \textsf{hDDF}(x, z), \textsf{hDDF}(y, z), \textsf{hMDF}(x, w), \\ 
&\textsf{mO}(z, w), \textsf{mO}(y, w)													& & \textsf{hMDF}(y, w), \textsf{wF}(x, v), \textsf{wF}(y, v)
\end{align*}
\caption{Summarized queries for Reactome (top left), Uniprot (top right), LUBM (bottom left) and UOBM (bottom right).}
\label{figure:queries}
\end{figure}

First, note that computing the restricted chase employing renaming techniques to deal with equality is way more efficient than computing the oblivious chase with singularization.
We conjecture that this is because the efficient built-in capabilities of RDF\-Ox to deal with equality and the fact that the rules that result from the application of singularization are rather cumbersome.
Second, see that our proposed algorithm is also superior to PAGOdA when it comes to CQ answering.
Third, the implementation of the restricted chase outperforms the DL reasoner Konclude  by an order of magnitude when it comes to solve materialization of the larger samples considered (note that, by computing the chase of a program we already solve materialization).
It is clear that our implementation also scales much better than the OWL DL reasoner.

\section{Conclusions and Future Work}
\label{conclusions}

We introduce a novel acyclicity notion for \hsriq TBoxes and prove it to be, theoretically and empirically, more general than previously existing conditions \cite{CG+13:acyclicity}.
To the best our knowledge, this is the first acyclicity notion (for ontologies or rules) which considers termination of the restricted chase algorithm.
Moreover, our contribution is also relevant in practice: Based on our ideas, we produce an implementation which vastly outperforms state-of-the-art reasoners.

As future work, we plan to lift our acyclicity condition to the case of general rules; i.e., not only those resulting from the translation of \hsriq TBoxes.
We also intend to work on further optimizing our implementation of the \RCA{n} membership check and our restricted chase based algorithm.

\medskip
\noindent 
\textbf{Acknowledgements. } 
We wish to thank Bernardo Cuenca Grau for extensive discussions on the subject and valuable feedback.
This work was supported by the National Science Foundation under awards 1017255 \emph{III: Small: TROn -- Tractable Reasoning with Ontologies} and 1440202 \emph{EarthCube Building Blocks: Collaborative Proposal: GeoLink -- Leveraging Semantics and Linked Data for Data Sharing and Discovery in the Geosciences}; the \emph{ERC grant 647289} and the \emph{European Research Council grant CODA 647289}.

\bibliographystyle{splncs03}
\bibliography{reference}



\end{document}